\documentclass[journal]{IEEEtran}

\usepackage[T1]{fontenc}
\usepackage{graphicx}
\usepackage{amsthm}
\usepackage{bm, bbm}
\usepackage{mathtools}
\usepackage{amsmath}
\usepackage{amssymb}

\usepackage{algorithm}
\usepackage{algorithmicx}
\usepackage{algpseudocode}

\usepackage{tikz}
\usetikzlibrary{positioning,calc,fit,backgrounds,shapes.geometric,decorations.pathreplacing}

\algrenewcommand\algorithmicrequire{\textbf{Input:}}
\algrenewcommand\algorithmicensure{\textbf{Output:}}

\newtheorem{lemma}{Lemma}

\newtheorem{example}{Example}

\usepackage{url}
\usepackage{cite}
\usepackage{float}
\usepackage{xcolor}
\usepackage{hyperref}

\usepackage{geometry}
\date{}

\begin{document}

\title{Learning to Decode Quantum LDPC Codes via Cluster-Based Sequential Belief Propagation}

\author{Mohsen~Moradi\textsuperscript{\href{https://orcid.org/0000-0001-7026-0682}{\includegraphics[scale=0.06]{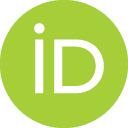}}},
Taejoon~Kim\textsuperscript{\href{https://orcid.org/0000-0002-4017-9530}{\includegraphics[scale=0.06]{figs/ORCID}}},
R\'emi~A.~Chou\textsuperscript{\href{https://orcid.org/0000-0003-4431-3175}{\includegraphics[scale=0.06]{figs/ORCID}}},
and  David~G.~M.~Mitchell\textsuperscript{\href{https://orcid.org/0000-0002-3544-9225}{ \includegraphics[scale=0.06]{figs/ORCID}}}
\thanks{Mohsen Moradi and Taejoon~Kim are with the School of Electrical, Computer and Energy Engineering, Arizona State University, Tempe, AZ 85287, USA (e-mail: mmorad11@asu.edu, taejoonkim@asu.edu).}
\thanks{R\'emi~Chou is with the Department of Computer Science and Engineering, 
The University of Texas at Arlington, Arlington, TX 76019, USA (e-mail: remi.chou@uta.edu).}
\thanks{David G. M. Mitchell is with the Klipsch School of Electrical and Computer Engineering, New Mexico State University, Las Cruces, NM 88003, USA (e-mail: dgmm@nmsu.edu).}%
\thanks{
This work is in part supported by the National Science Foundation (NSF) under grants CNS2451268, CNS2514415, ITE2515378, and CCF-2145917, and the Office of Naval Research (ONR) under Grant N000142112472.}
}

\maketitle

\begin{abstract}
Belief-propagation (BP) decoding for quantum low-density parity-check (QLDPC) codes is attractive due to its low complexity, but its performance is often limited by short cycles, degeneracy, and convergence failures. Recently, reinforcement-learning-based sequential variable-node (VN) scheduling (RL-S) was shown to improve BP decoding by learning state-dependent update orders. However, the VN-by-VN nature of that approach offers limited within-iteration parallelism, since only one VN is updated at a time. In this paper, we propose a cluster-based extension of RL-S for QLDPC codes. The VNs are partitioned into fixed clusters, and at each scheduling step the RL agent selects one cluster to update, after which all VNs in that cluster are updated in parallel using the same pre-update incoming messages. To keep the tabular state space practical for large cluster sizes, we introduce a permutation-invariant cluster state based on a normalized histogram of local mismatch weights, followed by quantization. This representation makes the number of cluster states depend on the quantization resolution rather than the cluster size. We also develop the corresponding cluster-level Markov decision process, reward function, and Q-learning update. Numerical results on representative QLDPC codes show that the proposed clustered learned scheduling preserves most of the error-rate benefit of VN-level learned sequential scheduling while substantially reducing the number of scheduling decisions per BP iteration, thereby providing an attractive latency-parallelism tradeoff.
\end{abstract}


\section{Introduction}
\label{sec:intro}

Quantum error correction is indispensable for large-scale quantum information processing because physical qubits are inherently noisy and unprotected computations accumulate errors rapidly \cite{shor1995scheme,calderbank1996good,steane1996error}. Among the many code families proposed for fault tolerance, quantum low-density parity-check (QLDPC) codes have emerged as a particularly promising direction. Their sparse stabilizer structure supports low-weight syndrome extraction, while recent code constructions have significantly improved the achievable rate-distance tradeoff, including hypergraph-product codes \cite{tillich2014quantum}, balanced-product codes \cite{breuckmann2021balanced}, lifted-product / almost-linearly distanced constructions  \cite{panteleev2022almostlinear}, quantum Tanner codes \cite{leverrier2022quantum}, and finite-length families such as bivariate bicycle (BB) codes \cite{bravyi2024high}. As a result, QLDPC codes are now widely viewed as a serious candidate for low-overhead fault-tolerant quantum architectures.
These advances make the decoder design problem increasingly important: good codes are now available, but practical performance depends critically on whether one can decode them both accurately and quickly.

While recent code constructions have greatly strengthened the case for QLDPC codes, their practical performance is still tightly linked to the design of efficient and reliable decoders.
Belief propagation (BP) is an attractive starting point because it has low complexity, relies only on local message passing on a sparse Tanner graph, and is highly successful for classical LDPC codes. For QLDPC codes, however, plain BP is well known to perform poorly \cite{poulin2008iterative,roffe2020landscape}. Two difficulties are especially important. First, the commutativity constraints of stabilizer codes induce many short cycles in the decoding graph, violating the independence assumptions of BP. Second, quantum degeneracy means that many distinct physical error patterns can produce the same syndrome and may even correspond to equivalent recovery operations on the codespace. In practice, these effects create symmetric pseudo-codewords, trapping structures, and oscillatory decoder dynamics, so conventional BP often fails to converge or converges to a logically incorrect solution \cite{raveendran2021trapping,kuo2022exploiting}.

A substantial literature has therefore developed around improving BP-based decoding for QLDPC codes. One major line of work augments BP with a post-processing stage. 
BP with ordered-statistics decoding (BP--OSD) \cite{panteleev2021degenerate,roffe2020landscape} improves BP by converting BP reliabilities into a syndrome-consistent solution through ordered-statistics post-processing.
BP with stabilizer inactivation (BP--SI) reduces the cost of this post-processing step by restricting the linear solve to a smaller set of carefully selected variables \cite{ducrest2022si}. 
More recent variants such as check-agnosia, localized statistics decoding (LSD), and low-complexity syndrome-based linear-programming post-processing seek to retain strong performance while improving parallelizability or lowering latency \cite{ducrest2024checkagnosia,hillmann2025lsd,javed2026sblp}. These methods can be very effective, but they introduce an auxiliary stage beyond plain BP whose latency, control complexity, or hardware footprint can become significant. 
A second line of work aims to improve the message-passing process itself rather than appending a heavy post-processor. Examples include memory-enhanced BP that explicitly exploits degeneracy \cite{kuo2022exploiting}, guided decimation \cite{yao2024bpgd,alinia2025decimation}, layered and random-order schedules \cite{ducrest2023layered}, graph-neural or hybrid learned decoders \cite{maan2025astra}, and recent relay-style BP refinements for real-time decoding under circuit-level noise \cite{muller2025relaybp}. The common theme is that the poor performance of plain BP is not simply a matter of insufficient iterations; rather, one must break the harmful symmetries created by short cycles and degeneracy, while preserving as much of BP's locality and low complexity as possible.

Among the many proposed improvements to BP-based QLDPC decoding, scheduling is particularly attractive. Changing the order of message updates does not alter the code itself, does not require a separate post-processing stage, and can often be implemented with relatively low overhead. In classical coding, serial and shuffled BP schedules are known to improve convergence, and reinforcement learning (RL) has also been used to learn effective update orders \cite{habib2023gldpc,moradi2025polar}. For QLDPC codes, sequential check-node (CN) and variable-node (VN) schedules have likewise been shown to outperform conventional flooding BP \cite{moradi2026sequential}. More recently, reinforcement-learning-based sequential VN scheduling (RL-S) was introduced, where the next VN is selected from a local syndrome-driven state, yielding substantial gains over both flooding and random sequential schedules while preserving the underlying BP message-update rule \cite{moradi2026rlbp}. In particular, RL-S can achieve significantly better error-correction performance than conventional BP and, for a fixed iteration cap, often succeeds in much fewer average iterations.
Subsequent extensions further addressed the latency--performance tradeoff of learned sequential BP under depolarizing noise. A list-based extension showed that the learned sequential decoder can achieve significantly better error-correction performance while maintaining the same decoding latency \cite{moradi2026rlls}. A complementary learned bit-flipping extension improved performance by selecting promising flip candidates after an initial learned sequential decoding stage and processing the resulting continuation branches in parallel, thereby improving the latency--performance tradeoff without adding a heavy post-processing step \cite{moradi2026rlsbf}.
 Taken together, these results show that, for QLDPC decoding, the scheduling policy is an important part of the decoder design itself.

In this paper, we address a central limitation of the sequential decoding approach. In RL-S, the decoder updates one VN at a time. This VN-level control is beneficial because it injects state-dependent asymmetry into BP and improves convergence, but it also creates a long serial decision chain: one BP iteration requires up to $n$ scheduling decisions for a length $n$ code.
As a result, although RL-S often converges in fewer average iterations and attains substantially better decoding performance than conventional BP, its VN-by-VN schedule can still lead to a less favorable latency profile from an implementation perspective. In particular, the fully serial scheduling process limits within iteration parallelism and can increase the latency, even when the underlying local message updates are themselves simple. 
Put differently, RL-S offers an attractive convergence-performance tradeoff, but at a scheduling granularity that may be too fine for latency-sensitive implementations.

Our goal is to preserve the main advantage of learned sequential scheduling while reducing its serial depth. To this end, we propose a \emph{cluster-based} extension of RL-S for QLDPC codes. The VNs are partitioned into fixed clusters. At each scheduling step, the RL agent selects one cluster, and all VNs in that cluster are then updated in parallel using the same pre-update incoming messages. The resulting hard-decision changes are applied simultaneously. This produces a decoder that is still sequential at the \emph{cluster} level, but parallel inside each selected cluster. The key design challenge is then no longer the BP update itself, but how to represent the state of a cluster compactly enough for tabular RL to remain practical.

A technical idea of this paper is to represent each cluster in a permutation-invariant manner through the distribution of local mismatch weights across its VNs. Rather than storing the ordered states of all VNs in a cluster, which would cause the state space to grow rapidly with cluster size, our proposed approach summarizes each cluster by a normalized histogram of local mismatch weights and then quantizes that histogram. In this way, the number of admissible cluster states depends mainly on the quantization resolution rather than on the cluster size itself. The resulting representation is compact, invariant to arbitrary reorderings of VNs within the cluster, and well matched to the syndrome-driven nature of sequential BP decoding.

Our proposed clustered RL-S decoder retains state-dependent scheduling at a coarser level while enabling parallel processing inside each chosen cluster. In that sense, it provides a direct latency--parallelism tradeoff inside the learned-scheduling framework.
We formulate the corresponding cluster-level Markov decision process (MDP), derive the reward and Q-learning update, and characterize how a parallel cluster update changes the residual mismatch pattern.
Through numerical results on representative QLDPC codes, we show that our proposed clustered decoder preserves much of the error-rate improvement of VN-level learned sequential scheduling while requiring substantially fewer scheduling decisions per BP iteration. For the \( [[288,12,18]] \) BB code \cite{bravyi2024high} over the depolarizing channel, our results show that clustered RL-S with cluster size \(B=10\), corresponding to a cluster-level scheduling depth of \(29\) per iteration, can achieve error-correction performance close to that of VN-level RL-S, which requires \(288\) scheduling decisions per iteration. Thus, under an ideal parallel cluster-update model and with sufficient quantization resolution, our proposed decoder can reduce the serial scheduling latency by approximately one order of magnitude.

The remainder of the paper is organized as follows. Section~\ref{sec:background} reviews the stabilizer and CSS framework, syndrome-based decoding under the independent Pauli-\(X\) channel, and the role of message-passing schedules in BP decoding. Section~\ref{sec:clustered} develops the proposed clustered RL-S formulation. Section~\ref{sec:complexity_latency} discusses the complexity and latency of BP, RL-S, and clustered RL-S. Section~\ref{sec:numerical} presents numerical results, and Section~\ref{sec:conclusion} concludes the paper.

\section{Background}
\label{sec:background}

\subsection{Stabilizer and CSS quantum LDPC codes}
Quantum stabilizer codes are the quantum counterpart of classical linear block codes \cite{calderbank1996good,steane1996error}.
Let $\mathcal{P}_n$ denote the $n$-qubit Pauli group, consisting of $n$-fold tensor products of the single-qubit Pauli operators $I$, $X$, $Y$, and $Z$, with phases in $\{\pm1,\pm i\}$. 
An $[[n,k,d]]$ stabilizer code is defined by an abelian subgroup $\mathcal{S}\subset\mathcal{P}_n$ generated by $n-k$ independent commuting Hermitian Pauli operators and satisfying $-I_{2^n}\notin\mathcal{S}$. 
The codespace is the joint $+1$ eigenspace of the elements of $\mathcal{S}$ and has dimension $2^k$.

A particularly important subclass is the Calderbank--Shor--Steane (CSS) family \cite{calderbank1996good,steane1996error}, for which the stabilizer generators can be separated into $X$-type and $Z$-type checks. A CSS code is specified by two binary parity-check matrices
\[
H_X\in\mathbb{F}_2^{m_X\times n},
\qquad
H_Z\in\mathbb{F}_2^{m_Z\times n},
\]
satisfying the orthogonality condition
\[
H_X H_Z^\top = 0 \quad \text{over } \mathbb{F}_2,
\]
which is the binary representation of stabilizer commutativity. A CSS code is called a QLDPC code when these matrices are sparse, i.e., when every row and every column has relatively small Hamming weight compared with the blocklength.

\subsection{Syndrome decoding and degeneracy under the independent Pauli-\texorpdfstring{$X$}{X} channel}
In this paper we first focus on the independent Pauli-$X$ channel for conceptual clarity. In this model, each physical qubit is flipped independently with probability $p_x$, and only bit-flip errors are present. The physical error can therefore be represented by a binary vector
\[
\bm{e}\in\mathbb{F}_2^n,
\]
where $e_i=1$ indicates an $X$ error on qubit $i$.

For this $X$-only setting, only the relevant $Z$-type parity checks matter. We denote by
\[
H_1 \in \mathbb{F}_2^{m\times n}
\]
the binary parity-check matrix used for this component of decoding, and the measured syndrome is
\[
\bm{s} = H_1 \bm{e} \in \mathbb{F}_2^m.
\]
A decoder outputs an estimate $\hat{\bm{e}}$ based on $\bm{s}$ and the channel parameter $p_x$.

Unlike in ordinary classical syndrome decoding, recovering the exact physical error is not the real objective. In a quantum code, different physical errors may induce the same syndrome and may even act identically on the codespace up to multiplication by stabilizers. This is the phenomenon of \emph{degeneracy} \cite{poulin2008iterative,kuo2022exploiting}. Consequently, successful decoding only requires that the residual operator between the true error and the chosen correction be logically trivial; exact physical identification is unnecessary. This feature is fundamentally beneficial from the viewpoint of quantum information, but it complicates BP decoding because the syndrome alone does not isolate a unique most-likely binary error pattern.

At the level of binary syndrome equations, all vectors in the affine set
\[
\{\bm{x}\in\mathbb{F}_2^n \mid H_1\bm{x} = \bm{s}\}
\]
are syndrome-consistent. However, these candidates need not be equally useful from the viewpoint of logical recovery. In simulations, one therefore typically checks not only syndrome consistency but also whether the residual error induces a nontrivial logical action. For BP-based decoders, degeneracy manifests itself as ambiguity in the soft information: multiple nearby candidates may explain the same syndrome, leading to weak or oscillatory beliefs rather than a clean separation between correct and incorrect local decisions \cite{poulin2008iterative,raveendran2021trapping}.

\subsection{Tanner graph and belief-propagation decoding}
The binary matrix $H_1$ defines a Tanner graph with VNs
$v_1,\ldots,v_n$ and check nodes $c_1,\ldots,c_m$. For
$i\in\{1,\ldots,n\}$ and $j\in\{1,\ldots,m\}$, VN $v_i$ is adjacent
to check node $c_j$ if and only if $H_1(j,i)=1$. For each VN $v_i$,
define
\[
\mathcal N(v_i)=\left\{c_j \mid H_1(j,i)=1\right\},
\]
and for each check node $c_j$, define
\[
\mathcal N(c_j)=\left\{v_i \mid H_1(j,i)=1\right\}.
\]

For the independent Pauli-$X$ channel, the a priori log-likelihood ratio (LLR) at each VN is
\[
\lambda_i \triangleq \log \frac{1-p_x}{p_x}.
\]
In standard sum-product BP, variable-to-check and check-to-variable messages are iteratively updated according to local rules on the Tanner graph. Denoting by $m_{v_i\to c_j}^{(\ell)}$ and $m_{c_j\to v_i}^{(\ell)}$ the VN-to-CN and CN-to-VN messages associated with edge $(v_i,c_j)$ at iteration $\ell$, the CN update takes the form

{\small
\begin{equation}
m_{c_j\to v_i}^{(\ell)}
=
2\tanh^{-1}\!\left(
(-1)^{s_j}
\prod_{v_{i'}\in \mathcal{N}(c_j)\setminus \{v_i\}}
\tanh\!\left(\frac{m_{v_{i'}\to c_j}^{(\ell-1)}}{2}\right)
\right),
\label{eq:bp_cn_update_background}
\end{equation}
}
while the VN update is
\begin{equation}
m_{v_i\to c_j}^{(\ell)}
=
\lambda_i+
\sum_{c_{j'}\in \mathcal{N}(v_i)\setminus \{c_j\}}
m_{c_{j'}\to v_i}^{(\ell)}.
\label{eq:bp_vn_update_background}
\end{equation}
The posterior LLR at VN $v_i$ is then
\begin{equation}
L_i^{(\ell)}
=
\lambda_i+
\sum_{c_j\in\mathcal N(v_i)}
m_{c_j\to v_i}^{(\ell)},
\label{eq:bp_posterior_background}
\end{equation}
and the hard decision is
\begin{equation}
\hat e_i^{(\ell)} =
\begin{cases}
0, & L_i^{(\ell)} > 0,\\[1mm]
1, & L_i^{(\ell)} \le 0.
\end{cases}
\label{eq:bp_hard_decision_background}
\end{equation}
The resulting estimated syndrome is $H_1\hat{\bm e}^{(\ell)}$.

For QLDPC codes, the issue is that the Tanner graph contains many short cycles; information can therefore return quickly to its origin and be repeatedly overcounted.
In addition, degeneracy implies that local reliabilities need not concentrate cleanly on a single candidate error pattern. This is why plain flooding BP often stalls or oscillates on QLDPC instances even when its classical counterpart is effective \cite{poulin2008iterative,roffe2020landscape,kuo2022exploiting}.

For later use, it is convenient to define the residual mismatch vector associated with a current hard estimate $\hat{\bm e}$ as
\begin{equation}
\bm{\delta} \triangleq \bm{s} \oplus (H_1\hat{\bm e}),
\label{eq:residual_mismatch_background}
\end{equation}
where $\oplus$ denotes componentwise addition modulo $2$. Thus, $\delta_j=1$ indicates that check node $c_j$ is currently unsatisfied.

\subsection{Sequential scheduling, learned scheduling, and the motivation for clustering}
The conventional BP decoder uses a \emph{flooding} schedule: all CN messages are updated in parallel from the previous VN messages, and then all VN messages are updated in parallel from the new CN messages. A \emph{sequential} or \emph{serial} schedule updates nodes one at a time and can immediately reuse fresher messages. Such schedules are well known to help classical BP decoding, and related schedule-optimization ideas have also been explored with RL \cite{habib2023gldpc,moradi2025polar}.

For QLDPC decoding, scheduling is especially relevant because the decoding graph has many short cycles. In \cite{moradi2026sequential} it is shown that sequential VN scheduling (SVNS) and sequential CN scheduling (SCNS) can already improve performance over flooding BP. Building on that observation, \cite{moradi2026rlbp} introduced a learned SVNS (RL-S) decoder in which, within each BP iteration, an RL agent chooses the next VN to update based on a local syndrome-driven state. After a VN is selected, its incident messages and local hard decision are updated, the residual mismatch is refreshed locally, and the VN is removed from the remaining set for that BP iteration. This can substantially reduce non-convergence. 

The advantage of this VN-level learned schedule is also its main drawback: one VN is processed at a time. Hence, for a length $n$ code, a single BP iteration may require up to $n$ separate scheduling decisions. From a hardware perspective, this creates a long serial control path and limits within-iteration parallelism. The clustered decoder proposed in this paper is designed precisely to alleviate that limitation. Instead of selecting a single VN, the RL agent selects a cluster of VNs, after which all VNs in the chosen cluster are updated in parallel from the same pre-update information. The challenge is then to preserve the essential syndrome-driven information needed for learned scheduling while moving from a node-level to a cluster-level action space. The next section develops that formulation.

\section{Clustered RL Scheduling}\label{sec:clustered}

The VN-by-VN learned schedule in RL-S improves convergence by injecting controlled asymmetry into BP, but it offers limited within-iteration parallelism because only one VN is updated at a time. To reduce the scheduling depth while retaining a learned, state-dependent update order, we group VNs into clusters and let the RL agent choose which cluster to process next. Once a cluster is selected, all VNs in that cluster are updated in parallel using the same pre-update incoming messages, and the resulting hard-decision changes are then applied simultaneously. In this way, the decoder preserves the main idea of learned scheduling, but operates at a coarser and more hardware-friendly granularity.

Fix an integer cluster size $B$ and partition the VN set
$\{v_1,\ldots,v_n\}$ into
\[
N_{\mathrm{cl}} \triangleq \left\lceil \frac{n}{B}\right\rceil
\]
disjoint clusters
\[
\mathcal{C}_1,\ldots,\mathcal{C}_{N_{\mathrm{cl}}},
\qquad
\mathcal{C}_a \subseteq \{v_1,\ldots,v_n\},
\]
such that
\[
\mathcal{C}_a \cap \mathcal{C}_{a'} = \emptyset \quad (a\neq a'),
\qquad
\bigcup_{a=1}^{N_{\mathrm{cl}}} \mathcal{C}_a = \{v_1,\ldots,v_n\}.
\]
For each cluster $\mathcal C_a$, define its associated index set as
\[
\mathcal I_a=\{i \mid v_i\in\mathcal C_a\}.
\]
In our numerical results, this partition is drawn once at initialization and then kept fixed throughout both training and testing. Specifically, for each of the first $N_{\mathrm{cl}}-1$ clusters, we assign $B$ VNs uniformly at random without replacement, and the final cluster contains the remaining VNs.
The random VN partitions used in this work are not designed to
guarantee conflict-free hardware implementations. Accordingly, the
latency analysis assumes ideal parallel cluster updates. Related work
on classical LDPC decoding constructs conflict-free CN clusters
\cite{ozkan2026clustered}; extending such graph-aware partitioning
methods to VN clusters for QLDPC decoding is left for future work.

\subsection{Cluster-level MDP formulation}

The cluster-level Markov decision process (MDP) follows the VN-based RL-S decoder, except that the actions now correspond to clusters rather than individual VNs. An episode samples a crossover probability $p_x$ from the training set, generates an error pattern and its syndrome, initializes the BP beliefs and messages, and then runs clustered RL-S decoding for at most $I_{\max}$ BP iterations.

The environment state is the full decoder configuration,
\[
\bigl\{m_{v_i\to c_j},m_{c_j\to v_i}:(v_i,c_j)\in E\bigr\},
\quad \{L_i\},\quad \hat{\bm e},\quad \bm{\delta},
\]
where \(E\) denotes the edge set of the Tanner graph.
Inside each BP iteration, the agent maintains a remaining set
\[
\mathcal{R}_t \subseteq \{1,\ldots,N_{\mathrm{cl}}\}
\]
of clusters that have not yet been visited in that iteration. Choosing action $a_t\in\mathcal{R}_t$ means that cluster $\mathcal{C}_{a_t}$ is selected and all VNs in that cluster are updated in parallel. After the update, the selected cluster is removed from the remaining set. Thus, as in the VN-based decoder, the schedule is applied without replacement within each BP iteration.

\subsection{Cluster state}\label{sec:clusterState}

A direct state description that stores the ordered collection of all VN-level states inside a cluster is unnecessarily large and depends on the arbitrary ordering of the VNs within that cluster. To obtain a compact and permutation-invariant representation, we summarize the local mismatch weights inside each cluster.

For each VN $v_i$, let
\[
A_i \triangleq |\mathcal{N}(v_i)|,
\qquad
A_{\max} \triangleq \max_i A_i.
\]
Using a fixed deterministic ordering of the neighboring checks of VN $v_i$, write
\[
\mathcal{N}(v_i)=(c_{j_1},c_{j_2},\ldots,c_{j_{A_i}}).
\]
We define the local binary mismatch vector \cite{moradi2026rlbp}
\[
\bm{b}_i
\triangleq
[\delta_{j_1},\delta_{j_2},\ldots,\delta_{j_{A_i}},0,\ldots,0]
\in \{0,1\}^{A_{\max}},
\]
where the vector is zero-padded to length $A_{\max}$ when $A_i<A_{\max}$.

From this local pattern, we define the local mismatch weight
\[
\omega_i \triangleq \|\bm{b}_i\|_1,
\]
which is simply the number of unsatisfied neighboring checks seen by VN $v_i$.

For a cluster $\mathcal{C}_a$, we aggregate these local mismatch weights into the histogram
\[
c_r(a)
\triangleq
\bigl|\{v_i\in\mathcal{C}_a \mid \omega_i = r\}\bigr|,
\qquad
r\in\{0,1,\ldots,A_{\max}\}.
\]
These counts satisfy
\[
\sum_{r=0}^{A_{\max}} c_r(a)=|\mathcal{C}_a|.
\]
The corresponding unquantized cluster state is
\begin{equation}
\bm{c}_a
\triangleq
\bigl(c_0(a),c_1(a),\ldots,c_{A_{\max}}(a)\bigr).
\label{eq:raw_cluster_state}
\end{equation}
Note that the size of $\bm{c}_a$ still grows with the cluster size. To make the state representation comparable across different cluster sizes and to control the state-space growth, we instead normalize the histogram by the cluster size as
\[
\rho_r(a)
\triangleq
\frac{c_r(a)}{|\mathcal{C}_a|},
\qquad
r=0,1,\ldots,A_{\max},
\]
so that
\[
\sum_{r=0}^{A_{\max}} \rho_r(a)=1.
\]

To control the state-space size, we choose a quantization resolution \(L\). We quantize
\(\rho(a)\) onto the integer simplex using a largest-remainder rule. First, define
\[
\tilde q_r(a)=L\rho_r(a),~~
b_r(a)=\lfloor \tilde q_r(a)\rfloor,~~
f_r(a)=\tilde q_r(a)-b_r(a).
\]
Let
\[
D=L-\sum_{r=0}^{A_{\max}} b_r(a).
\]
Then \(q_r(a)\) is obtained by setting \(q_r(a)=b_r(a)\) and adding one to the
\(D\) indices with the largest fractional parts \(f_r(a)\), with ties broken by
increasing \(r\). Thus,
\begin{equation}
q_r(a)\in\mathbb Z_{\ge 0},\qquad
\sum_{r=0}^{A_{\max}}q_r(a)=L .
\label{eq:simplex_quantization}
\end{equation}
The resulting quantized cluster state is
\begin{equation}
\bm{\sigma}_a
\triangleq
\bigl(q_0(a),q_1(a),\ldots,q_{A_{\max}}(a)\bigr).
\label{eq:quantized_cluster_state}
\end{equation}
This representation is permutation-invariant within the cluster, captures the distribution of local mismatch weights, and, crucially, makes the number of possible states depend on the quantization resolution $L$ rather than the cluster size.

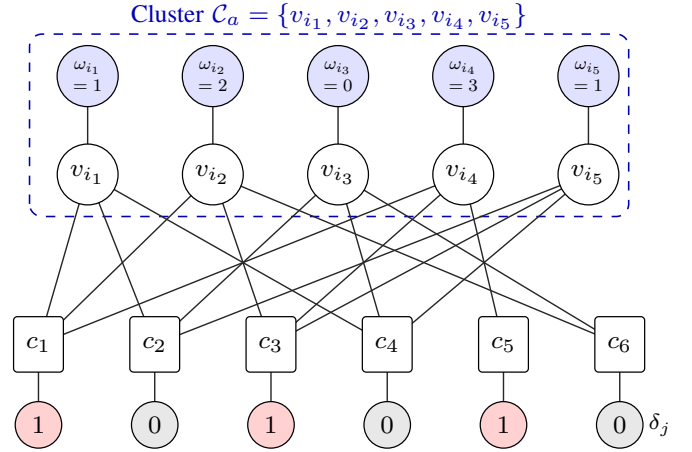
\begin{figure}[t]
\centering
\resizebox{\columnwidth}{!}{\begin{tikzpicture}[
    font=\footnotesize,
    x=1cm,
    y=1cm,
    edge/.style={
        line width=0.45pt,
        black!85
    },
    vn/.style={
        draw,
        circle,
        fill=white,
        inner sep=0pt,
        minimum size=6.8mm
    },
    cn/.style={
        draw,
        rounded corners=1.2pt,
        fill=white,
        inner xsep=4pt,
        inner ysep=2.2pt,
        minimum height=6.2mm
    },
    wnode/.style={
        draw,
        circle,
        fill=blue!12,
        inner sep=-.5pt,
        minimum size=5.8mm,
        font=\footnotesize,
        align=center
    },
    delta0/.style={
        draw,
        circle,
        fill=gray!18,
        inner sep=0pt,
        minimum size=5.2mm
    },
    delta1/.style={
        draw,
        circle,
        fill=red!18,
        inner sep=0pt,
        minimum size=5.2mm
    }
]

\coordinate (V1) at (0.55,3.05);
\coordinate (V2) at (1.95,3.05);
\coordinate (V3) at (3.35,3.05);
\coordinate (V4) at (4.75,3.05);
\coordinate (V5) at (6.15,3.05);

\coordinate (W1) at (0.55,4.15);
\coordinate (W2) at (1.95,4.15);
\coordinate (W3) at (3.35,4.15);
\coordinate (W4) at (4.75,4.15);
\coordinate (W5) at (6.15,4.15);

\coordinate (C1) at (0.00,1.15);
\coordinate (C2) at (1.30,1.15);
\coordinate (C3) at (2.60,1.15);
\coordinate (C4) at (3.90,1.15);
\coordinate (C5) at (5.20,1.15);
\coordinate (C6) at (6.50,1.15);

\coordinate (D1) at (0.00,0.25);
\coordinate (D2) at (1.30,0.25);
\coordinate (D3) at (2.60,0.25);
\coordinate (D4) at (3.90,0.25);
\coordinate (D5) at (5.20,0.25);
\coordinate (D6) at (6.50,0.25);

\draw[
    blue!70!black,
    dashed,
    rounded corners=4pt,
    line width=0.5pt
] (-0.10,2.58) rectangle (6.60,4.62);

\node[
    font=\footnotesize,
    text=blue!70!black
] at (3.25,4.82)
{Cluster $\mathcal{C}_a=\{v_{i_1},v_{i_2},v_{i_3},v_{i_4},v_{i_5}\}$};

\draw[edge] (W1) -- (V1);
\draw[edge] (W2) -- (V2);
\draw[edge] (W3) -- (V3);
\draw[edge] (W4) -- (V4);
\draw[edge] (W5) -- (V5);


\draw[edge] (V1) -- (C1);
\draw[edge] (V1) -- (C2);
\draw[edge] (V1) -- (C4);

\draw[edge] (V2) -- (C1);
\draw[edge] (V2) -- (C3);
\draw[edge] (V2) -- (C6);

\draw[edge] (V3) -- (C2);
\draw[edge] (V3) -- (C4);
\draw[edge] (V3) -- (C6);

\draw[edge] (V4) -- (C1);
\draw[edge] (V4) -- (C3);
\draw[edge] (V4) -- (C5);

\draw[edge] (V5) -- (C2);
\draw[edge] (V5) -- (C3);
\draw[edge] (V5) -- (C4);

\draw[edge] (C1) -- (D1);
\draw[edge] (C2) -- (D2);
\draw[edge] (C3) -- (D3);
\draw[edge] (C4) -- (D4);
\draw[edge] (C5) -- (D5);
\draw[edge] (C6) -- (D6);


\node[wnode] at (W1) {\scalebox{0.7}{\(\begin{array}{c}\omega_{i_1}\\[1pt]=1\end{array}\)}};
\node[wnode] at (W2) {\scalebox{0.7}{\(\begin{array}{c}\omega_{i_2}\\[1pt]=2\end{array}\)}};
\node[wnode] at (W3) {\scalebox{0.7}{\(\begin{array}{c}\omega_{i_3}\\[1pt]=0\end{array}\)}};
\node[wnode] at (W4) {\scalebox{0.7}{\(\begin{array}{c}\omega_{i_4}\\[1pt]=3\end{array}\)}};
\node[wnode] at (W5) {\scalebox{0.7}{\(\begin{array}{c}\omega_{i_5}\\[1pt]=1\end{array}\)}};

\node[vn] at (V1) {$v_{i_1}$};
\node[vn] at (V2) {$v_{i_2}$};
\node[vn] at (V3) {$v_{i_3}$};
\node[vn] at (V4) {$v_{i_4}$};
\node[vn] at (V5) {$v_{i_5}$};

\node[cn] at (C1) {$c_1$};
\node[cn] at (C2) {$c_2$};
\node[cn] at (C3) {$c_3$};
\node[cn] at (C4) {$c_4$};
\node[cn] at (C5) {$c_5$};
\node[cn] at (C6) {$c_6$};

\node[delta1] at (D1) {$1$};
\node[delta0] at (D2) {$0$};
\node[delta1] at (D3) {$1$};
\node[delta0] at (D4) {$0$};
\node[delta1] at (D5) {$1$};
\node[delta0] at (D6) {$0$};

\node[font=\scriptsize] at (6.95,0.25) {$\delta_j$};

\end{tikzpicture}}
\caption{Tanner subgraph used in Example~\ref{ex:1}. The highlighted cluster
\(\mathcal{C}_a=\{v_{i_1},\ldots,v_{i_5}\}\) is connected to checks
\(c_1,\ldots,c_6\) according to the listed neighbor sets. The values below the check nodes are the
current residual mismatches \(\delta_j\), and the circles above the variable nodes show the induced
local mismatch weights \(\omega_{i_k}\).}
\label{fig:cluster_example1_tanner}
\end{figure}

\begin{example}\label{ex:1}
    
Suppose $A_{\max}=3$ and a cluster $\mathcal{C}_a$ of size $B=5$ contains
\[
\mathcal{C}_a=\{v_{i_1},v_{i_2},v_{i_3},v_{i_4},v_{i_5}\}.
\]
Assume that the current residual mismatch vector is
\[
\bm{\delta}=(\delta_1,\delta_2,\delta_3,\delta_4,\delta_5,\delta_6)
=(1,0,1,0,1,0),
\]
so that checks $c_1,c_3,c_5$ are currently unsatisfied and
checks $c_2,c_4,c_6$ are satisfied.

Assume further that the neighboring checks of the VNs in
$\mathcal{C}_a$ are, in the fixed deterministic order used by the decoder,
\[
\mathcal{N}(v_{i_1})=(c_1,c_2,c_4),\qquad
\mathcal{N}(v_{i_2})=(c_1,c_3,c_6),
\]
\[
\mathcal{N}(v_{i_3})=(c_2,c_4,c_6),\qquad
\mathcal{N}(v_{i_4})=(c_1,c_3,c_5),
\]
\[
\mathcal{N}(v_{i_5})=(c_2,c_3,c_4).
\]
The Tanner subgraph corresponding to this setup is illustrated in
Fig.~\ref{fig:cluster_example1_tanner}. The circles below the check nodes
show the current residual mismatch values \(\delta_j\), while the circles
above the VNs show the resulting local mismatch weights \(\omega_{i_k}\)
computed below.

The corresponding local mismatch vectors are obtained by
restricting $\bm{\delta}$ to the neighboring checks of each VN as
\[
\bm{b}_{i_1}=[\delta_1,\delta_2,\delta_4]=[1,0,0],\qquad
\bm{b}_{i_2}=[\delta_1,\delta_3,\delta_6]=[1,1,0],
\]
\[
\bm{b}_{i_3}=[\delta_2,\delta_4,\delta_6]=[0,0,0],\qquad
\bm{b}_{i_4}=[\delta_1,\delta_3,\delta_5]=[1,1,1],
\]
\[
\bm{b}_{i_5}=[\delta_2,\delta_3,\delta_4]=[0,1,0].
\]
Their local mismatch weights are therefore
\[
\omega_{i_1}=1,\qquad
\omega_{i_2}=2,\qquad
\omega_{i_3}=0,\qquad
\omega_{i_4}=3,~~
\omega_{i_5}=1.
\]
Hence the histogram counts are
\[
c_0(a)=1,\qquad
c_1(a)=2,\qquad
c_2(a)=1,\qquad
c_3(a)=1.
\]
If we use the raw histogram representation, the resulting cluster state is
\[
\bm{\sigma}_a
=
\bigl(c_0(a),c_1(a),c_2(a),c_3(a)\bigr)
=
(1,2,1,1).
\]
In words, the cluster contains one VN with no unsatisfied
neighboring checks, two VNs with one unsatisfied neighboring
check, one VN with two unsatisfied neighboring checks, and one
VN with three unsatisfied neighboring checks.

If instead we use the normalized representation, then
\[
(\rho_0(a),\rho_1(a),\rho_2(a),\rho_3(a))
=
\left(\frac{1}{5},\frac{2}{5},\frac{1}{5},\frac{1}{5}\right).
\]
For example, if the quantization resolution is \(L=8\), then
\[
L\rho(a)
=
(1.6,3.2,1.6,1.6).
\]
The integer floor vector is therefore
\[
(b_0(a),b_1(a),b_2(a),b_3(a))
=
(1,3,1,1),
\]
and the corresponding fractional remainders are
\[
(f_0(a),f_1(a),f_2(a),f_3(a))
=
(0.6,0.2,0.6,0.6).
\]
Since
\[
D
=
L-\sum_{r=0}^{3} b_r(a)
=
2,
\]
we must add one unit to the two largest fractional remainders. The largest
remainders are tied at \(r=0,2,3\); with ties broken by increasing \(r\), the
two selected positions are \(r=0\) and \(r=2\). Hence,
\[
(q_0(a),q_1(a),q_2(a),q_3(a))
=
(2,3,2,1),
\]
which satisfies
\[
q_0(a)+q_1(a)+q_2(a)+q_3(a)=8.
\]

\end{example}

\subsection{State-space size under quantization}

For a fixed $A_{\max}$ and quantization resolution $L$, the number of possible quantized cluster states is the number of nonnegative integer solutions of
\[
q_0+q_1+\cdots+q_{A_{\max}}=L,
\]
namely
\begin{equation}\label{eq:num_states_L}
N_{\mathrm{states}}^{(L)}
=
\binom{L+A_{\max}}{A_{\max}}.
\end{equation}
Hence the tabular state space no longer scales with the cluster size $B$. Note that without quantization, the number of possible cluster states is
\begin{equation}\label{eq:num_states_B}
N_{\mathrm{states}}^{(B)}
=
\binom{B+A_{\max}}{A_{\max}}.
\end{equation}

For the $[[882, 24, 18\leq d \leq 24]]$ code B1 \cite{panteleev2021degenerate}, we have $A_{\max}=3$, so
\[
N_{\mathrm{states}}^{(L)}=\binom{L+3}{3}.
\]
As a concrete example, when $L=8$,
\[
N_{\mathrm{states}}^{(8)}=\binom{11}{3}=165.
\]
Therefore, if $n=882$ and $B=20$, then $N_{\mathrm{cl}}=45$ and the Q-table contains
\[
165\times 45 = 7425
\]
entries. If instead $B=60$, then $N_{\mathrm{cl}}=15$ and the Q-table contains
\[
165\times 15 = 2475
\]
entries. In both cases, the cluster-state resolution is controlled by $L$, not by $B$.

\subsection{Parallel cluster update}\label{subsec:cluster_update}

When a cluster $\mathcal{C}_a$ is selected, all VNs in that cluster are updated
in parallel. More precisely, each VN $v_i\in\mathcal{C}_a$ first computes its
updated belief and tentative hard decision using the \emph{pre-update} incoming
messages. After all tentative decisions in the cluster have been formed, the
induced hard-decision flips are applied simultaneously. This makes the
within-cluster update independent of any arbitrary ordering of the VNs inside
the cluster.

Using the residual mismatch $\bm{\delta}$,
suppose that the selected cluster update changes the hard decisions
on the index set $\mathcal{F}\subseteq\mathcal{I}_a$.
If $\bm{u}_i$ denotes the $i$th standard basis vector, then the
updated hard estimate is
\[
\hat{\bm e}'
=
\hat{\bm e}
\oplus
\left(
\bigoplus_{i\in\mathcal{F}}\bm{u}_i
\right).
\]
The following lemma describes how this simultaneous set of flips changes the
residual mismatch.

\begin{lemma}\label{lem:delta_cluster}
The updated residual mismatch
\[
\bm{\delta}'=\bm{s}\oplus(H_1\hat{\bm e}')
\]
satisfies, for each check index $j$,
\[
\delta'_j
=
\delta_j
\oplus
\left(\sum_{i\in\mathcal{F}} H_1(j,i)\right)\!\!\!\!\pmod{2}.
\]
Equivalently, check node $c_j$ changes its residual status if and only if it is adjacent
to an odd number of flipped VNs in $\mathcal{F}$.
\end{lemma}

\begin{proof}
By linearity over $\mathbb{F}_2$,
\[
H_1\hat{\bm e}'
=
H_1\hat{\bm e}
\oplus
\bigoplus_{i\in\mathcal{F}} H_1\bm{u}_i .
\]
Therefore,
\[
\bm{\delta}'
=
\bm{s}\oplus H_1\hat{\bm e}'
=
\bm{s}\oplus H_1\hat{\bm e}
\oplus
\bigoplus_{i\in\mathcal{F}} H_1\bm{u}_i
=
\bm{\delta}
\oplus
\bigoplus_{i\in\mathcal{F}} H_1\bm{u}_i .
\]
Taking the $j$th component gives
\[
\delta'_j
=
\delta_j
\oplus
\bigoplus_{i\in\mathcal{F}} H_1(j,i)
=
\delta_j
\oplus
\left(\sum_{i\in\mathcal{F}} H_1(j,i)\right)\!\!\!\!\pmod{2}.
\]
Hence, check node $c_j$ changes its residual status exactly when it is connected to an odd number of flipped VNs.
\end{proof}

Lemma~\ref{lem:delta_cluster} shows that the effect of a parallel cluster update on the residual mismatch is completely local: only checks adjacent to at least one flipped VN in the selected cluster can change, and each such check changes according to an odd-parity rule.

\subsection{Reward and Q-learning update for cluster actions}\label{subsec:qlearn_cluster}

For the the current mismatch weight \(w = \|\bm{\delta}\|_1\), suppose that at step $t$ the agent selects cluster $a_t\in\mathcal{R}_t$, where $\mathcal{R}_t$ is the set of clusters that remain available in the current BP iteration. Let
$w_{\mathrm{before}}$
and
$w_{\mathrm{after}}$
denote the mismatch weights immediately before and after applying the corresponding parallel cluster update.
We define the one-step reward as the normalized mismatch reduction
\begin{equation}\label{eq:cluster_reward}
r_t
\triangleq
\frac{w_{\mathrm{before}}-w_{\mathrm{after}}}
{\sum_{i\in\mathcal{I}_{a_t}} A_i},
\qquad
A_i \triangleq |\mathcal{N}(v_i)|.
\end{equation}
We add a terminal bonus of $+1$ when the current estimate becomes
syndrome-consistent, i.e., 
\[
r_t \leftarrow r_t + 1
\qquad \text{if } w_{\mathrm{after}} = 0.
\]

After selecting $a_t$, we remove that action from the remaining set, i.e.,
\[
\mathcal{R}_{t+1}
\triangleq
\mathcal{R}_t \setminus \{a_t\}.
\]
Let $\bm{\sigma}_a$ denote the current quantized cluster state associated with action $a$. The one-step lookahead value is
\begin{equation}\label{eq:best_future_cluster}
V_{t+1}
\triangleq
\max_{a'\in\mathcal{R}_{t+1}} Q(\bm{\sigma}_{a'},a'),
\end{equation}
with the convention that $V_{t+1}=0$ when $\mathcal{R}_{t+1}=\emptyset$.

The tabular Q-learning update is then
\begin{equation}\label{eq:qlearn_cluster_update}
Q(\bm{\sigma}_{a_t},a_t)
\leftarrow
Q(\bm{\sigma}_{a_t},a_t)
+
\alpha\Bigl(
r_t + \gamma V_{t+1} - Q(\bm{\sigma}_{a_t},a_t)
\Bigr),
\end{equation}
where $\alpha\in(0,1]$ is the learning rate and $\gamma\in[0,1)$ is the discount factor.

\subsection{Extension to the depolarizing channel}
\label{subsec:depolarizing_extension}

The clustered RL-S framework can also be applied to the depolarizing channel by using the same quaternary BP representation as in the VN-level RL-S decoder \cite{moradi2026rlbp}. In this setting, each qubit error is a Pauli symbol
\[
q_i\in\{I,X,Y,Z\},
\]
with
\[
\Pr(q_i=I)=1-p,
\]
\[
\Pr(q_i=X)=\Pr(q_i=Y)=\Pr(q_i=Z)=\frac{p}{3},
\]
where \(p\) denotes the physical depolarizing error probability per qubit.
For a Pauli symbol \(q\), define its binary components as
\[
e^X(q)\triangleq \mathbbm{1}[q\in\{X,Y\}],
\qquad
e^Z(q)\triangleq \mathbbm{1}[q\in\{Y,Z\}].
\]
Thus, the \(X\)-component is detected by the \(Z\)-type checks, while the
\(Z\)-component is detected by the \(X\)-type checks. 
For a CSS code, the residual mismatch vectors are
\[
\bm{\delta}^X
\triangleq
\bm{s}^X\oplus H_X\hat{\bm e}^Z,
\qquad
\bm{\delta}^Z
\triangleq
\bm{s}^Z\oplus H_Z\hat{\bm e}^X,
\]
where the superscripts on \(\bm{\delta}^X\) and \(\bm{\delta}^Z\)
label the \(X\)- and \(Z\)-type check syndromes, respectively. Thus,
\(\bm{\delta}^X\) is associated with the \(X\)-type checks and the
estimated \(Z\)-component, while \(\bm{\delta}^Z\) is associated with
the \(Z\)-type checks and the estimated \(X\)-component. The total mismatch
weight is
\[
w
\triangleq
\|\bm{\delta}^X\|_1+\|\bm{\delta}^Z\|_1.
\]

The cluster state is formed in the same way as in the independent Pauli-\(X\)
case, except that each VN now observes unsatisfied checks from both Tanner
graphs. To distinguish the two CSS Tanner graphs, let \(c_a^X\) denote the
\(a\)th \(X\)-type check node, corresponding to row \(a\) of \(H_X\), and let
\(c_b^Z\) denote the \(b\)th \(Z\)-type check node, corresponding to row \(b\)
of \(H_Z\). For qubit \(i\), define
\[
\mathcal{N}_X(v_i)
\triangleq
\{c_a^X \mid H_X(a,i)=1\},
\]
\[
\mathcal{N}_Z(v_i)
\triangleq
\{c_b^Z \mid H_Z(b,i)=1\}.
\]
The depolarizing-channel local mismatch weight is then
\[
\omega_i
\triangleq
\sum_{c_a^X\in\mathcal{N}_X(v_i)} \delta^X_a
+
\sum_{c_b^Z\in\mathcal{N}_Z(v_i)} \delta^Z_b .
\]
Thus, \(\omega_i\) counts the number of currently unsatisfied \(X\)- and
\(Z\)-type checks adjacent to qubit \(i\), using the residual mismatch vectors
\(\bm{\delta}^X\) and \(\bm{\delta}^Z\). The cluster histogram is then
constructed from the values of \(\omega_i\) inside each cluster and quantized
exactly as in Section~\ref{sec:clusterState}. Hence, the same
permutation-invariant cluster-state representation applies.

When a cluster is selected, all qubits in that cluster are updated in parallel
using the quaternary BP update rule. If the Pauli decision at qubit \(i\) changes
from \(\hat q_i\) to \(\hat q'_i\), define
\[
\Delta e_i^X
\triangleq
e^X(\hat q_i)\oplus e^X(\hat q'_i),
\qquad
\Delta e_i^Z
\triangleq
e^Z(\hat q_i)\oplus e^Z(\hat q'_i).
\]
Then the residual mismatches are updated as
\[
\bm{\delta}^{X\prime}
=
\bm{\delta}^X
\oplus
H_X\Delta\bm e^Z,
\qquad
\bm{\delta}^{Z\prime}
=
\bm{\delta}^Z
\oplus
H_Z\Delta\bm e^X.
\]
Thus, the depolarizing-channel version uses the same cluster-selection,
reward, and Q-learning structure, with the binary residual mismatch replaced by the pair \((\bm{\delta}^X,\bm{\delta}^Z)\). The full quaternary BP message-update details generalize the RL-S decoder in \cite{moradi2026rlbp}; here, the learned component only changes the scheduling action from a single VN to a cluster of VNs.

\subsection{Training and greedy inference}

During training, the action is selected using an $\epsilon$-greedy policy over the current remaining set $\mathcal{R}_t$. At training episode $\nu$, the decoder selects an action uniformly at random from $\mathcal{R}_t$ with probability $\epsilon_\nu$. Otherwise, with probability $1-\epsilon_\nu$, it selects an action that maximizes $Q(\sigma_a,a)$ over all $a\in\mathcal{R}_t$. Ties between maximizing actions are broken deterministically.
In our implementation, the exploration probability decreases linearly according to
\[
\epsilon_\nu
=
\max\left\{
\epsilon_{\min},
\epsilon_0
\left(
1-\frac{\nu-1}{E_{\max}-1}
\right)
\right\},
~
\nu\in\{1,\ldots,E_{\max}\}.
\]
After training, the learned policy is used greedily by setting
$\epsilon=0$. Thus, at scheduling step $t$, the selected action satisfies
\[
a_t
\in
\arg\max_{a\in\mathcal{R}_t} Q(\sigma_a,a).
\]
Thus, one BP iteration requires only $N_{\mathrm{cl}}$ scheduling decisions rather than $n$, while each decision activates a full cluster of VNs in parallel. This is precisely the latency--parallelism tradeoff introduced by clustered RL-S: as $B$ increases, the number of scheduling decisions per iteration decreases, but each action becomes coarser.

The training and inference procedures for the independent Pauli-\(X\) setting are summarized in
Algorithms~\ref{alg:cluster_qlearning} and~\ref{alg:cluster_inference},
respectively. Algorithm~\ref{alg:cluster_qlearning} describes the offline
Q-learning procedure used to learn the cluster-selection table, while
Algorithm~\ref{alg:cluster_inference} describes the greedy decoding procedure
used after training. In both algorithms, the main difference from VN-level RL-S
is that the action is a cluster index rather than a single VN index. The
depolarizing-channel implementation follows the same structure, with the binary
hard decision and residual mismatch replaced by the quaternary decision and the
pair \((\bm{\delta}^X,\bm{\delta}^Z)\).

\begin{algorithm}[t]
\caption{Offline training of the clustered RL-S decoder for the independent Pauli-\(X\) channel}
\label{alg:cluster_qlearning}
\begin{algorithmic}[1]
\State Initialize Q-table $Q \leftarrow 0$
\For{episode $=1$ to $E_{\max}$}
    \State Sample $p_x$ uniformly from $\mathcal{P}_{\mathrm{train}}$
    \State Sample $\bm e \sim \mathrm{Bernoulli}(p_x)$ and compute $\bm s = H_1 \bm e$
    \State Set $\mu_x = \log\!\left(\frac{1-p_x}{p_x}\right)$ and initialize beliefs/messages
    \State Initialize $\hat{\bm e}$ and residual mismatch $\bm{\delta}=\bm s\oplus(H_1\hat{\bm e})$
    \For{BP-iteration $=1$ to $I_{\max}$}
        \State $\mathcal{R} \leftarrow \{1,2,\ldots,N_{\mathrm{cl}}\}$
        \For{step $=1$ to $N_{\mathrm{cl}}$}
            \If{$\bm{\delta}=\bm 0$}
                \State break
            \EndIf
            \State Compute quantized cluster states $\bm{\sigma}_a$ $\forall$ $a\in\mathcal{R}$
            \State Choose $a_t\in\mathcal{R}$ by $\epsilon$-greedy on $Q(\bm{\sigma}_a,a)$
            \State $w_{\mathrm{before}} \leftarrow \|\bm{\delta}\|_1$
            \State Apply one parallel cluster update on $\mathcal{C}_{a_t}$
            \State $w_{\mathrm{after}} \leftarrow \|\bm{\delta}\|_1$
            \State $r_t \leftarrow \dfrac{w_{\mathrm{before}}-w_{\mathrm{after}}}
            {\sum_{i\in\mathcal{I}_{a_t}} A_i}$
            \If{$w_{\mathrm{after}}=0$}
                \State $r_t \leftarrow r_t + 1$
            \EndIf
            \State $\mathcal{R} \leftarrow \mathcal{R}\setminus\{a_t\}$
            \If{$\mathcal{R}=\emptyset$}
                \State $V_{t+1} \leftarrow 0$
            \Else
                \State Recompute $\bm{\sigma}_{a'}$ for all $a'\in\mathcal{R}$
                \State $V_{t+1} \leftarrow \max_{a'\in\mathcal{R}} Q(\bm{\sigma}_{a'},a')$
            \EndIf
            \State $Q(\bm{\sigma}_{a_t},a_t) \leftarrow Q(\bm{\sigma}_{a_t},a_t)
            + \alpha\!\left(r_t + \gamma V_{t+1} - Q(\bm{\sigma}_{a_t},a_t)\right)$
        \EndFor
    \EndFor
\EndFor
\end{algorithmic}
\end{algorithm}

Algorithm~\ref{alg:cluster_qlearning} gives the offline training procedure for
the clustered RL-S decoder. Line~1 initializes the Q-table. Lines~2--6 generate
one training episode by sampling the channel parameter, generating an error and
its syndrome, initializing the BP messages and beliefs, and forming the initial
hard decision and residual mismatch. Lines~7--8 start the BP iterations and
initialize the set of available clusters for the current iteration. Lines~9--12
perform the cluster-selection steps within one BP iteration and stop early if
the residual mismatch has already become zero. Line~13 computes the current
quantized state of each remaining cluster. Line~14 selects a cluster using the
\(\epsilon\)-greedy rule. Lines~15--17 apply the parallel cluster update and
measure the change in the residual mismatch weight. Lines~18--21 compute the
one-step reward, including the terminal bonus upon syndrome convergence. Line~22
removes the selected cluster from the remaining set so that each cluster is
visited at most once per BP iteration. Lines~23--28 compute the one-step
lookahead value over the remaining clusters. Finally, line~29 updates the
Q-table entry associated with the selected cluster and its current quantized
state.

\begin{algorithm}[t]
\caption{Greedy inference of the clustered RL-S decoder for the independent Pauli-\(X\) channel}
\label{alg:cluster_inference}
\begin{algorithmic}[1]
\State Input: measured syndrome $\bm s$, learned table $Q$
\State Initialize BP messages, beliefs, hard decision $\hat{\bm e}$, and residual mismatch $\bm{\delta}=\bm s\oplus(H_1\hat{\bm e})$
\For{BP-iteration $=1$ to $T$}
    \State $\mathcal{R} \leftarrow \{1,2,\ldots,N_{\mathrm{cl}}\}$
    \For{step $=1$ to $N_{\mathrm{cl}}$}
        \If{$\bm{\delta}=\bm 0$}
            \State return $\hat{\bm e}$
        \EndIf
        \State Compute quantized cluster states $\bm{\sigma}_a$ for all $a\in\mathcal{R}$
        \State Select
        \[
        a_t \in \arg\max_{a\in\mathcal{R}} Q(\bm{\sigma}_a,a)
        \]
        \State Apply one parallel cluster update on $\mathcal{C}_{a_t}$
        \State $\mathcal{R} \leftarrow \mathcal{R}\setminus\{a_t\}$
    \EndFor
\EndFor
\State Return final hard decision $\hat{\bm e}$
\end{algorithmic}
\end{algorithm}

Algorithm~\ref{alg:cluster_inference} gives the greedy decoding procedure used
after training. Line~1 takes the measured syndrome and the learned Q-table as
inputs. Line~2 initializes the BP messages, beliefs, hard decision, and residual
mismatch. Lines~3--4 start the BP iterations and reset the remaining cluster set
at the beginning of each iteration. Lines~5--8 process the clusters within the
current iteration and return the current estimate as soon as the residual
mismatch becomes zero. Line~9 computes the quantized states of the remaining
clusters. Line~10 selects the cluster with the largest learned Q-value. Line~11
updates all VNs in the selected cluster in parallel, and line~12 removes that
cluster from the remaining set. If no zero residual mismatch is found within the
iteration budget, line~15 returns the final hard decision.

\section{Complexity and Latency}
\label{sec:complexity_latency}

In this section, we compare BP, RL-S, and clustered RL-S algorithms in terms of the number of local VN updates and the serial decoding depth. We assume a sparse QLDPC Tanner graph with bounded node degrees and a fixed maximum number of iterations \(I_{\max}\). Then:
\begin{itemize}
    \item For conventional flooding BP, all VN updates within one iteration can be performed in parallel. Therefore, under ideal parallel hardware, the latency scales as $O(I_{\max}).$
Since each iteration updates all VNs and their incident messages, the total computational complexity scales as $O(I_{\max} n)$;
\item For RL-S, the decoder updates one VN at a time. Hence, one BP iteration requires up to \(n\) sequential VN updates. The total number of local updates is still \(O(I_{\max}n)\), so the computational complexity remains comparable to BP. However, the serial decoding depth scales as $O(I_{\max} n),$ because the VN updates are performed one by one;
\item Finally, for clustered RL-S, the VNs are partitioned into
\[
N_{\mathrm{cl}}=\left\lceil \frac{n}{B}\right\rceil
\]
clusters. During one BP iteration, the decoder selects clusters sequentially, but all VNs inside the selected cluster are updated in parallel. Therefore, the total number of local VN updates per iteration remains \(O(n)\), and the overall computational complexity is still $O(I_{\max} n).$
The main difference is in latency: the serial decoding depth now depends on the number of clusters rather than the blocklength. Thus, under the parallel cluster-update model, the latency scales as
$O(I_{\max} N_{\mathrm{cl}})$,
instead of \(O(I_{\max} n)\) for VN-level RL-S.
\end{itemize}  

The quantized cluster-state representation also keeps the Q-table size practical. The Q-table contains
\(N_{\mathrm{cl}}N_{\mathrm{states}}^{(L)}\) entries, where
\(N_{\mathrm{states}}^{(L)}\) is given in
\eqref{eq:num_states_L} for the independent Pauli-\(X\) decoder.
Thus, the allocated Q-table size is controlled by the quantization
resolution and the number of clusters, rather than directly by the
number of VNs in each cluster.

\section{Numerical Results}
\label{sec:numerical}

In this section, we evaluate the error-correction performance of our proposed clustered RL-S decoder. The main goal of these experiments is to study how much of the performance gain of VN-level learned sequential scheduling can be preserved when the scheduling action is moved from individual VNs to clusters of VNs. We report the block error rate as a function of the physical error probability.
A block is counted as erroneous if the decoder either fails to reach syndrome consistency or returns a syndrome-consistent correction whose residual action is logically nontrivial.

Throughout this section, \(B\) denotes the cluster size and \(L\) denotes the quantization resolution used in the normalized cluster-state representation introduced in Section~\ref{sec:clustered}. For a code of blocklength \(n\), the number of clusters is
\[
N_{\mathrm{cl}}=\left\lceil \frac{n}{B}\right\rceil .
\]
Thus, larger \(B\) reduces the number of scheduling decisions per BP iteration, while each scheduling action updates a larger group of VNs. The parameter \(L\) controls the resolution of the quantized histogram state.
The learned policies are trained offline by sampling the physical error probability uniformly from
\[
\mathcal{P}_{\mathrm{train}}=\{0.03,0.04,0.05,0.06,0.07\},
\]
and are then used greedily during inference. In the depolarizing-channel
comparisons, QBP denotes quaternary BP, and QBPGD denotes quaternary BP
with guided decimation.

\begin{figure}[t]
  \centering
  \includegraphics[width=1\linewidth]{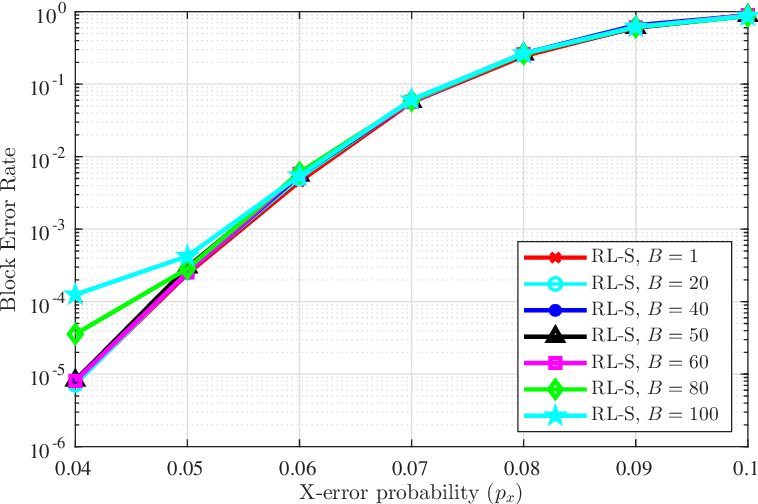}
  \caption{Block error rate of clustered RL-S decoding for the  \( [[882,24,18\le d\le 24]] \) code B1 over the independent Pauli-\(X\) channel. Different cluster sizes \(B\) are compared under the same maximum iteration budget \(T=100\).}
  \label{fig:B1_RL_SVNS_cluster}
\end{figure}

Fig.~\ref{fig:B1_RL_SVNS_cluster} shows the block error rate performance of clustered RL-S on the \( [[882,24,18\le d\le 24]] \) code B1 \cite{panteleev2021degenerate} over the independent Pauli-\(X\) channel. In this experiment, the maximum number of BP iterations is fixed to \(T=100\), and we compare several cluster sizes where we use the unquantized histogram state $\bm{c}_a$ in \eqref{eq:raw_cluster_state}. This figure studies the effect of replacing VN-by-VN scheduling with cluster-based scheduling. The results show that moderate clustering does not significantly degrade the decoding performance. In particular, increasing the cluster size up to about \(B=60\) preserves  the error rate behavior of the VN-level learned sequential schedule. This is important from a latency perspective. For \(B=60\), the \(n=882\) VNs are partitioned into
\[
N_{\mathrm{cl}}=\left\lceil \frac{882}{60}\right\rceil=15
\]
clusters, where the last cluster contains \(42\) VNs. Hence, one BP iteration requires only \(15\) scheduling decisions instead of \(882\) VN-level decisions. Therefore, the serial scheduling depth per iteration is reduced substantially, while the error-correction performance remains close to that of the VN-level learned schedule for the moderate cluster sizes shown in the figure.

\begin{figure}[t]
  \centering
  \includegraphics[width=1\linewidth]{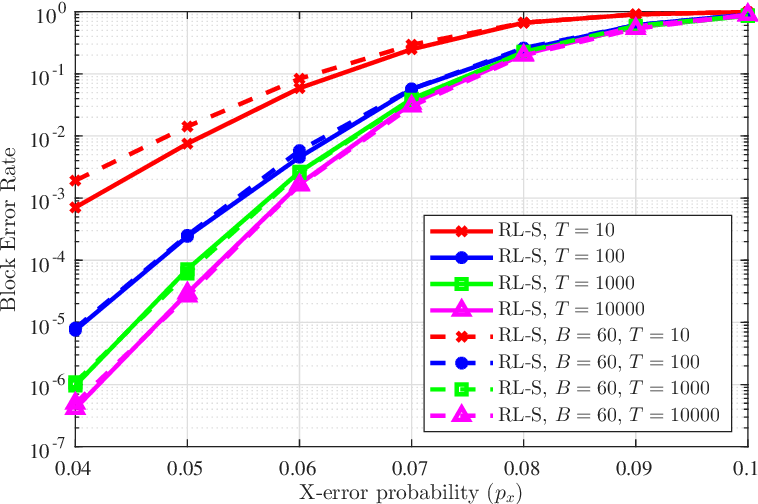}
  \caption{Block error rate for the  \( [[882,24,18\le d\le 24]] \) code B1 over the independent Pauli-\(X\) channel. The unclustered RL-S decoder is compared with clustered RL-S using cluster size \(B=60\) for several iteration budgets.}
  \label{fig:B1_RL_SVNS_cluster60}
\end{figure}

Fig.~\ref{fig:B1_RL_SVNS_cluster60} compares the unclustered RL-S decoder with the clustered decoder using \(B=60\) for different maximum iteration budgets with the code B1. For small iteration budgets, the clustered decoder can be more sensitive to the cluster update structure (assuming sufficiently large clusters), since each selected cluster applies several VN updates in parallel and therefore provides fewer opportunities to adapt the update order within one sweep. However, for iteration budgets of \(T=100\) and above, the clustered decoder closely follows the block error rate of the unclustered RL-S decoder. This indicates that our proposed cluster-level schedule preserves the main convergence benefit of learned sequential scheduling when the iteration budget is large enough.

\begin{figure}[t]
  \centering
  \includegraphics[width=1\linewidth]{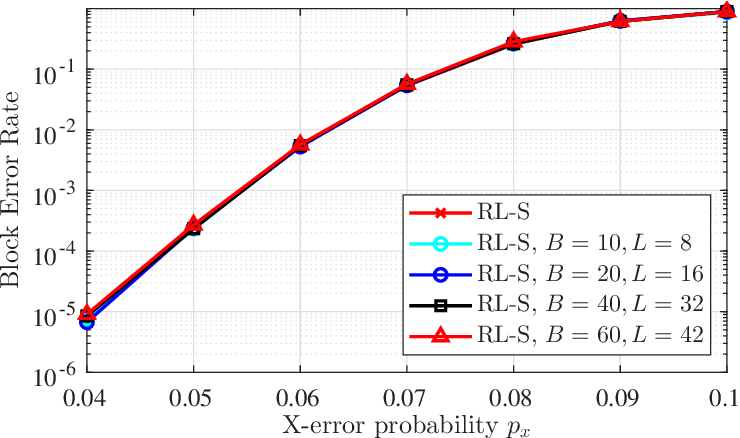}
  \caption{Effect of the quantization resolution on clustered RL-S decoding for the $[[882,24,18\leq d\leq24]]$ code B1 over the independent Pauli-$X$ channel. The iteration budget is fixed at $T=100$, while $B$ and $L$ are varied. The RL-S baseline ($B=1$) uses the original exact VN-level state, whereas all $B>1$ curves use the quantized histogram state.}
  \label{fig:RLS_B1_882_FER_Quantize}
\end{figure}

Fig.~\ref{fig:RLS_B1_882_FER_Quantize} evaluates our proposed quantized cluster-state representation on the code B1. Here, the cluster state is formed by the normalized histogram of local mismatch weights inside the cluster, followed by quantization with resolution \(L\).
Fig.~\ref{fig:B1_RL_SVNS_cluster60} uses the unquantized cluster
state $\bm{c}_a$ in \eqref{eq:raw_cluster_state}, which retains the
exact number of VNs having each local mismatch weight. In contrast,
Fig.~\ref{fig:RLS_B1_882_FER_Quantize} uses the quantized state
$\bm{\sigma}_a$ in \eqref{eq:quantized_cluster_state}, obtained by
normalizing the histogram and applying the simplex quantization in
\eqref{eq:simplex_quantization}; increasing \(L\) preserves finer
distinctions between normalized cluster histograms.
The results show that across the tested pairs of \(B\) and \(L\), the clustered decoders remain close to the unclustered RL-S baseline while using far fewer scheduling decisions per BP iteration. This suggests that the distribution of local mismatch weights provides a useful summary of how strongly the cluster is connected to the current residual syndrome.

\begin{figure}[t]
  \centering
  \includegraphics[width=1\linewidth]{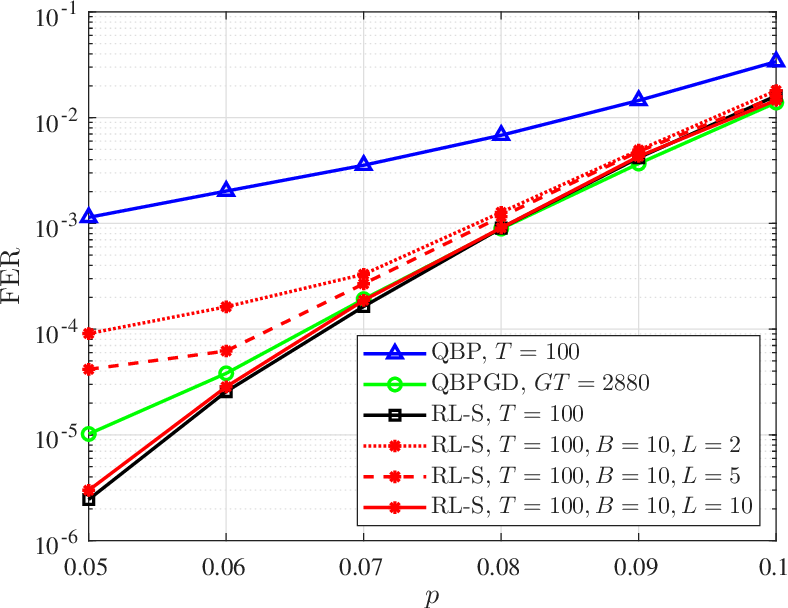}
  \caption{Block error rate comparison for the \( [[288,12,18]] \) BB code over the depolarizing channel. The proposed quantized clustered RL-S decoder is compared with QBP, QBPGD, and the unclustered RL-S decoder. The RL-S and clustered RL-S decoders use maximum iteration budget \(T=100\), while QBPGD uses \(288\) guided-decimation rounds with \(T=10\) BP iterations per round, i.e., total budget \(GT=2880\).}
  \label{fig:RLS_ClusterQuantized_BB288_FER}
\end{figure}

Fig.~\ref{fig:RLS_ClusterQuantized_BB288_FER} considers the \( [[288,12,18]] \) BB code \cite{bravyi2024high} over the depolarizing channel. In this experiment, our proposed quantized clustered RL-S decoder is compared with QBP, QBPGD, and the unclustered RL-S decoder. The QBPGD decoder uses \(288\) rounds of guided decimation with \(T=10\) BP iterations per round, corresponding to a total BP-iteration budget of \(2880\). The clustered learned schedules continue to provide a clear improvement over conventional QBP. The clustered decoder also remains competitive with the unclustered RL-S decoder. In particular, for \(L=10\), the error-correction performance is very close to the RL-S performance.

\section{Conclusions}
\label{sec:conclusion}

In this paper, we proposed a cluster-based extension of reinforcement-learning-based sequential BP decoding for QLDPC codes. Our proposed decoder reduces the serial scheduling depth of RL-S by selecting clusters of VNs rather than individual VNs. Once a cluster is selected, all VNs in that cluster are updated in parallel using the same pre-update information.
We also introduced a compact permutation-invariant cluster state based on the normalized histogram of local mismatch weights. After quantization with resolution \(L\), this representation keeps the tabular state space practical. We formulated the corresponding cluster-level MDP, reward function, Q-learning update, and parallel cluster-update rule.
Numerical results indicate that cluster-level learned scheduling is a promising approach for improving the latency--parallelism tradeoff of BP-based QLDPC decoders. Future work includes studying graph-aware cluster partitions, optimizing \(B\) and \(L\) for different QLDPC codes, and extending the numerical evaluation to additional channel models.



\bibliographystyle{IEEEtran}
\bibliography{bibliography}

\end{document}